\documentclass[journal]{IEEEtran}
\usepackage[utf8]{inputenc} 
\usepackage[T1]{fontenc}
\usepackage{url}              
\usepackage{cite}             

\usepackage[cmex10]{amsmath}  
\usepackage{mleftright}       
\mleftright                   

\usepackage{graphicx}         
\usepackage{booktabs}         

\usepackage{blindtext}
\usepackage{hyperref}

\usepackage{amsthm}

\usepackage{amssymb}
\usepackage{mathtools}

\usepackage{algorithmic}
\usepackage{algorithm}

\usepackage{comment}

\usepackage{amsmath,amsfonts,amsthm,bm} 

\newcommand{\E}{\mathbb{E}}
\newcommand{\Prob}{\mathbb{P}}

\newcommand{\diag}{\mathrm{diag}}

\newcommand{\LR}{\bold{R}^{\frac{1}{2}}}
\newcommand{\RT}{\bold{T}^{\frac{1}{2}}}
\newcommand{\MS}{\bold{S}^{\frac{1}{2}}}
\newcommand{\LD}{{\bold{D}^{\frac{1}{2}}}}

\newcommand{\FR}{{\bold{R}}}
\newcommand{\FT}{{\bold{T}}}
\newcommand{\FS}{{\bold{S}}}
\newcommand{\BQ}{{\bold{Q}}}
\newcommand{\BH}{{\bold{H}}}
\newcommand{\BI}{{\bold{I}}}
\newcommand{\BX}{{\bold{X}}}

\newcommand{\BZ}{{\bold{Z}}}
\newcommand{\BY}{{\bold{Y}}}
\newcommand{\FW}{{\bold{W}}}
\newcommand{\BG}{{\bold{G}}}

\newcommand{\BU}{{\bold{U}}}

\newcommand{\BA}{{\bold{A}}}
\newcommand{\BB}{{\bold{B}}}
\newcommand{\BC}{{\bold{C}}}
\newcommand{\BD}{{\bold{D}}}

\newcommand{\BW}{{\bold{W}}}
\newcommand{\BM}{{\bold{M}}}

\newcommand{\BF}{{\bold{F}}}

\newcommand{\BO}{{\mathcal{O}}}

\DeclareMathOperator{\Tr}{Tr}

\newcommand{\RNum}[1]{\uppercase\expandafter{\romannumeral #1\relax}}
\newcommand\numberthis{\addtocounter{equation}{1}\tag{\theequation}}

\newtheorem{remark}{Remark}
\newtheorem{theorem}{Theorem}
\newtheorem{lemma}{Lemma}

\begin{document}

\title{URLLC in IRS-Aided MIMO Systems: Finite Blocklength Analysis and Design} 

\author{Xin Zhang,~\IEEEmembership{Graduate Student Member,~IEEE} and Shenghui Song,~\IEEEmembership{Senior Member,~IEEE}
\thanks{
The authors are with the Department of Electronic and Computer Engineering, The Hong Kong University of Science
and Technology, Hong Kong (e-mail: xzhangfe@connect.ust.hk, eeshsong@ust.hk). 

This work was supported by a grant from the NSFC/RGC Joint Research Scheme sponsored by the Research Grants Council of the Hong Kong Special Administrative Region, China and National Natural Science Foundation of China (Project No. N\_HKUST656/22) and the Shenzhen Science and Technology Innovation Committee under Grant SGDX20210823103201006.
}
}

\maketitle

\begin{abstract}
This paper investigates the ultra reliable and low latency communication (URLLC) performance of the IRS-aided MIMO system. The upper and lower bounds of the optimal average error probability (OAEP) for the coding rate within $\BO(\frac{1}{\sqrt{Mn}})$ of the capacity are derived, where $n$ and $M$ represent the blocklength and the number of transmit antennas, respectively. To achieve this goal, a new central limit theorem (CLT) for the mutual information density over the IRS-aided MIMO system is derived in the asymptotic regime where the block-length, the IRS size, and number of the antennas go to infinity with the same pace. The CLT is then utilized to derive the closed-form upper and lower bounds for the OAEP. Based on the analysis result, a gradient-based algorithm is proposed to minimize the lower bound of the OAEP by optimizing the phase shift of the IRS. Simulation results validate the fitness of the CLT and the effectiveness of the proposed algorithm in optimizing the theoretical bound, as well as the performance of practical LDPC code.
\end{abstract}


\section{Introduction}
The intelligent reflecting surface (IRS) is regarded as a promising technology to improve the efficiency of wireless communication systems~\cite{cui2014coding}. By tuning the phase shifts, IRSs can create favorable channels between the transceivers with low energy consumption due to its passive nature. Furthermore, IRSs can be easily deployed and integrated into conventional communication systems. These favorable properties have motivated many works on the performance evaluation and design of the IRS-aided systems~\cite{zhang2021large,moustakas2023reconfigurable,zhang2022outage}.  However, despite the importance of ultra reliable and low latency communication (URLLC) in future wireless communication systems, the analysis and design of IRS-aided systems with finite blocklength (FBL) is still in its infancy.


\subsection{Performance Evaluation of IRS-Aided MIMO Systems}
There have been engaging results on the performance evaluation and design of IRS-aided multiple-input-multiple-output (MIMO) systems in the infinite blocklength regime. In~\cite{zhang2021large,moustakas2023reconfigurable}, the ergodic rate of the IRS-aided MIMO system was derived by random matrix theory (RMT) and maximized by jointly optimizing the transmit covariance matrix and phase shifts. The outage probability of the IRS-aided MIMO system was given and optimized by the gradient descent algorithm in~\cite{zhang2022outage}. However, the infinite-blocklength analysis is not capable of characterizing the fundamental limits of systems with FBL and the associated design may not be able to achieve the desired URLLC performance.



\subsection{Finite Blocklength Analysis}
There are three fundamental metrics for the performance analysis with FBL, i.e., the packet error rate, the blocklength, and the coding rate. The analysis for the trade-off among the three metrics is a great challenge. In~\cite{polyanskiy2010channel}, considering the FBL, the conventional Shannon's coding rate was refined and the maximal channel coding rate was determined by a normal approximation with respect to the blocklength $n$. By evaluating the channel dispersion, the coding rate of the single-input single-output (SISO) fading channel was given in~\cite{polyanskiy2011scalar}. The FBL capacity of independent and identically distributed (i.i.d.) isotropic coherent MIMO fading channels was investigated in~\cite{collins2018coherent}. The closed-form upper and lower bounds for the optimal average error probability (OAEP) of quasi-static MIMO Rayleigh channels were derived by RMT in~\cite{hoydis2015second}. The upper and lower bounds for the OAEP of Rayleigh-product MIMO channels were given in~\cite{zhang2022second}, where the impact of the number of scatterers was also investigated. 

Existing works on the FBL analysis for the IRS-aided system focus on the SISO~\cite{xie2021user} and MISO~\cite{cheng2022robust,abughalwa2022finite} channels. To the best of the authors' knowledge, the FBL analysis for IRS-aided MIMO systems remains unavailable in the literature and the associated phase shift design has not been investigated. In this paper, we will investigate the OAEP of IRS-aided MIMO systems for rates within $\BO(\frac{1}{\sqrt{Mn}})$ of the ergodic capacity. Charactering the OAEP of IRS-aided MIMO systems turns out to be a difficult problem due to the complex structure of the channel matrix, which is a sum of a random matrix (direct channel) and a product of the random matrices (cascaded channel). To tackle this issue, we will leverage the Gaussian tools from RMT, which has been shown to be effective and widely used in the analysis of MIMO systems~\cite{hachem2008new,zhang2022asymptotic}.

Inspired by the information spectrum approach in~\cite{hayashi2009information,hoydis2015second}, bounding the OAEP can be converted to investigating the distribution of the mutual information density (MID). To this end, a central limit theorem (CLT) for the MID is set up and the closed-form upper and lower bounds for the OAEP are derived. The contributions of this paper are:

1. Assuming that the number of antennas, the size of the IRS, and the blocklength go to infinity with the same pace, we set up a CLT for the MID over IRS-aided MIMO channels for any codes with the equal energy constraint. Specifically, we show that the characteristic function of the MID converges to that of the Gaussian distribution by RMT and give the expressions for the asymptotic mean and variance in closed forms. 


2. Based on the CLT, we give explicit expressions for the upper and lower bounds of the OAEP. To increase the reliability of the IRS-aided MIMO system, we propose a gradient-based method to minimize the derived lower bound by optimizing the phase shift of the IRS. 



3. Numerical results show that the derived upper and lower bounds have a small gap, which validates the tightness of proposed bound. It is also observed that the packet error probability of the LDPC code and the lower bound of the OAEP have a similar slope, and both can be effectively decreased by the proposed phase shift optimization algorithm.




  
\textit{Notations:} The matrix and vector are represented by the bold, upper case letter and bold, lower case letter, respectively. $\bold{A}^{H}$ and $\|\BA \|$ denote the conjugate transpose and the spectral norm of $\bold{A}$, respectively. $[\BA]_{i,j}$ or $A_{ij}$ represent the $(i,j)$-th entry of $\bold{A}$. $(\cdot)^{*}$ represents the conjugate of a complex number. $\Tr\BA$ denotes the trace of $\BA$ and $\bold{I}_{N}$ denotes the $N$ by $N$ identity matrix. $\mathbb{C}^{N}$ and $\mathbb{C}^{M\times N}$ denote the $N$-dimensional vector space and the $M$-by-$N$ matrix space, respectively. $\mathcal{N}(0,1)$ and $\mathcal{CN}(0,1)$ represent the standard Gaussian distribution and circularly Gaussian distribution, respectively. $\Phi(x)$ represents the cumulative distribution function (CDF) of the standard Gaussian distribution. $a  \downarrow b$ represents that $a$ approaches $b$ from the right. The convergence in distribution is denoted by $\xrightarrow[N \rightarrow \infty]{\mathcal{D}}$. Given a set $\mathcal{S}$, $\mathrm{P}(\mathcal{S})$ denotes the set of probability measures with support of a subset of $\mathcal{S}$ and $\mathrm{supp}(\cdot)$ denotes the support operator. $\Prob(\cdot)$ represents the probability operator and $\E x$ denotes the expectation of $x$. Big-O and little-o notations are represented by $\BO$ and $o$, respectively.


\section{System Model and Problem Formulation}
\label{sec_mod}
\subsection{System Model}
Consider an IRS-aided MIMO system with an $M$-antenna base station (BS) transmitting signal to an $N$-antenna user. The received signal $\bold{y}_t\in \mathbb{C}^{N}$ (channel output) is represented as
\begin{equation}
\label{sig_mod}
\bold{y}_{t}={\BH}\bold{c}_{t}+\sigma\bold{w}_{t}, ~~t=1,2,...,n.
\end{equation}
where $\bold{c}_{t}\in\mathbb{C}^{M}$, $\bold{w}_{t} \in\mathbb{C}^{N}\sim \mathcal{CN}(0,\bold{I}_{N})$ and ${\BH}\in\mathbb{C}^{N\times M}$ denote the transmit signal (channel input), the AWGN at time $t$, and the channel matrix, respectively. $n$ represents the blocklength (the channel use required to transmit a codeword) and $\sigma^2$ is the variance of the AWGN. The quasi-statistic channel is assumed, where $\BH$ does not change in $n$ channel uses. We assume that the statistical channel state information (CSI) is available at the transmitter and the IRS, and the perfect CSI is available at the receiver. For ease of illustration, we introduce the following notations:
$\BC^{(n)}\in \mathbb{C}^{N\times n} =[\bold{c}_1,\bold{c}_2,...,\bold{c}_n]$, $\BY^{(n)}\in \mathbb{C}^{N\times n}=[\bold{y}_1,\bold{y}_2,...,\bold{y}_n]$, and $\BW^{(n)} \in\mathbb{C}^{N\times n}=[\bold{w}_1,\bold{w}_2,...,\bold{w}_n]$. Then, the maximal energy constraint and equal energy constraint (sphere constraint) are given by
\begin{equation}
\label{max_cons}
\begin{aligned}
\mathcal{S}^{n}=\{\BC^{(n)} \in\mathbb{C}^{M\times n }| \frac{\Tr\BC^{(n)}(\BC^{(n)})^{H}}{Mn} \le 1  \},
\end{aligned}
\end{equation}
and
\begin{equation}
\label{sph_cons}
\begin{aligned}
\mathcal{S}^{n}_{=}=\{\BC^{(n)} \in\mathbb{C}^{M\times n }| \frac{\Tr\BC^{(n)}(\BC^{(n)})^{H}}{Mn} = 1  \},
\end{aligned}
\end{equation}
respectively. It is obvious that the equal energy constraint is a subset of the maximal energy constraint. In this paper, we will derive the bounds for the OAEP with the maximal energy constraint, which will be bounded by the error rate with the equal energy constraint as shown in Lemma~\ref{bnd_err}.

\subsection{Channel Model}
In this paper, we consider the Rayleigh fading for the BS-user, BS-IRS, and IRS-user channels. The effective channel of the MIMO system aided by the IRS with $L$ reflecting elements is the sum of a two-hop Rayleigh channel and a single-hop Rayleigh channel. Due to the space limitation, we assume that the transmit antennas at the BS are well separated such that there is no correlation at the transmitter, and the general case will be reported in the extended version of this work.  Thus, the channel matrix is given by
\begin{equation}
\BH=\LR\BX\RT_{IRS}\bold{\Phi}\LR_{IRS}\BY+\LD\BU,
\end{equation}
where $\FT_{IRS}\in \mathbb{C}^{L\times L}$ and $\FR_{IRS}\in \mathbb{C}^{L\times L}$ represent the correlation matrices at the IRS, and $\BD\in \mathbb{C}^{N\times M}$ and $\FR \in \mathbb{C}^{N\times M}$ denote the receive correlation matrices of the BS-user and IRS-user link, respectively. $\BX \in \mathbb{C}^{N\times L}$, $\BY \in \mathbb{C}^{L\times M}$, and $\BU \in \mathbb{C}^{N \times M}$ are random matrices with i.i.d. circularly Gaussian entries. The phase shift matrix of the IRS is given by 
\begin{equation}
\bold{\Phi}\in \mathbb{C}^{L\times L} =\diag(\phi_1,\phi_2,...,\phi_L),
\end{equation}
with $\phi_l=e^{\jmath \theta_l},\theta_l \in [0,2\pi)$. In the following, we give the definitions of the metrics concerned in this paper. 
\subsection{Average Error probability}
\textit{Encoding Mapping}:
With the maximal power constraint, a $(P_e^{(n)},G_n)$-code for the model in~(\ref{sig_mod}) can be represented by the following mapping $f$,
\begin{equation}
f:\mathcal{G}_n \rightarrow \mathbb{C}^{M\times n}.
\end{equation}
$\BC_{m}^{(n)}=f(m)\in \mathcal{S}^{n}$ represents the transmit symbols and $m$ is uniformly distributed in $\mathcal{G}_n=\{1,2,..,G_n \}$. $\mathcal{C}_{n}=\{f(1),f(2),...,f(G_n)\}$ is the codebook. 

\textit{The decoder mapping}: The decoder mapping from the channel output $\BY^{(n)}$ to the message can be represented by
\begin{equation}
g:\mathbb{C}^{N\times n} \rightarrow \mathcal{G}_{n} \cup \{e\}.
\end{equation}
The mapping $g$ gives the decision of $\hat{m}=g(\BY^{(n)})$, where $\BY^{(n)}=\BH f(m)+\sigma \BW^{(n)}$ denotes the received block. The decoder picks the transmitted message $m$ if it is correctly decoded otherwise an error $e$ occurs. Since $m$ is assumed to be uniformly distributed, the \textit{average error probability} for a code $\mathcal{C}_n$ with blocklength $n$, encoder $f$, decoder $g$ and the input $M_n$ is given by
\begin{equation}
P_e^{(n)}(\mathcal{C}_n)=\Prob(\hat{m} \neq m),
\end{equation}
where the evaluation involves the randomness of $\BH$, $\BW^{(n)}$, and $m\in \mathcal{G}_n$. Given a rate $R$, the OAEP is given by
\begin{equation}
\label{pro_e_ori}
\begin{aligned}
 P_e^{(n)}(R)=\inf_{\mathcal{C}_n:\mathrm{supp}(\mathcal{C}_n)\subseteq \mathcal{S}^{n}}
\left\{ P_e^{(n)}(\mathcal{C}_n) | \frac{1}{nM}\log(|\mathcal{C}_n|) \ge R \right\}.
\end{aligned}
\end{equation}
Here $R$ represents the per-antenna rate of each transmitted symbol. The main results of this paper are based on following three assumptions.

 \textbf{A.1.}  (Large system asssumption) $0<\lim\inf\limits_{M \ge 1} \frac{M}{L} \le \frac{M}{L}  \le \lim \sup\limits_{M \ge 1} \frac{M}{L} <\infty$, $0<\lim \inf\limits_{M \ge 1}  \frac{M}{N} \le \frac{M}{N}  \le \lim \sup\limits_{M \ge 1}  \frac{M}{N} <\infty$, $0<\lim\inf\limits_{M \ge 1}  \frac{M}{n} \le \frac{M}{n}  \le \lim \sup\limits_{M \ge 1} \frac{M}{n} <\infty$.

\textbf{A.2.} $ \limsup\limits_{N\ge 1} \| \FR \| <\infty$, $\limsup\limits_{N \ge 1} \| \BD \| <\infty$, $\limsup\limits_{N\ge 1} \|\FS \|$.

\textbf{A.3.} $\limsup\limits_{N\ge 1} \frac{1}{M}\Tr\FR>0  $, $\limsup\limits_{N\ge 1} \frac{1}{M}\Tr\BD>0 $, and $\limsup\limits_{N\ge 1} \frac{1}{M}\Tr\FS>0$.

 \textbf{A.1} assumes that $M$, $N$, $L$, and $n$ increase to infinity with the same pace, which is widely used in the large system analysis. We denote the ratios $\alpha=\frac{N}{M}$, $\beta=\frac{M}{L}$, and $\tau=\frac{n}{M}$. $M  \xrightarrow[]{\alpha,\beta, \tau}\infty$ represents the asymptotic regime where $n$, $N$, $M$, and $L$ grow to infinity with the fixed ratios $\alpha$, $\beta$, and $\tau$.~\textbf{A.2} and~\textbf{A.3} guarantee that the rank of the correlation matrices is not ignorable compared to the system dimensions~\cite{zhang2022outage}.

\subsection{Optimal Average Error Probability}

In this work, we consider the rate $R$ within $\BO(\frac{1}{\sqrt{Mn}})$ of $\E C(\sigma^2)$, i.e., $R=\E C(\sigma^2)+\frac{r}{\sqrt{Mn}}$, where $r$ is called the~\textit{second-order coding rate}~\cite{hayashi2009information,hoydis2015second}. Given a second-order coding rate $r$, the OAEP for the concerned system with $M$ transmit antennas and blockength $n$ is given by~\cite{hoydis2015second}
\begin{equation}
\label{def_oaep}
\begin{aligned}
&\Prob_{e}(r| \alpha, \beta, \tau)
{=}\inf\limits_{ \{\mathcal{C}_n:\mathrm{supp}(\mathcal{C}_n)\subseteq \mathcal{S}^{n}\}_{n=1}^{\infty}}
\{ \limsup_{n  \xrightarrow[]{\alpha,\beta, \tau}\infty} P_e^{(n)}(\mathcal{C}_n)| 
\\
&\times
\liminf_{n  \xrightarrow[]{\alpha,\beta, \tau}\infty}\sqrt{nM}(\frac{1}{nM}\log(|\mathcal{C}_n|)- \E C(\sigma^2))\ge r\},
\end{aligned}
\end{equation}
where $C(\sigma^2)=\frac{1}{M}\log\det(\bold{I}_{N}+\frac{1}{\sigma^2}\BH\BH^{H})$ denotes the per antenna capacity. From~(\ref{pro_e_ori}) and~(\ref{def_oaep}), we can observe that for $r=\BO(1)$, the rate $R=\E C(\sigma^2) +\frac{r}{\sqrt{Mn}}$ is a $\BO(\frac{1}{\sqrt{Mn}})$ perturbation around $\E C(\sigma^2)$.

Unfortunately, deriving the closed-form expression of the OAEP for arbitrary systems dimensions $M$, $L$, $N$, and coding rate $R$ is very challenging. To handle the challenge brought by the complex structure of the IRS-aided MIMO channels, we will back off from the infinity by assuming that assumptions~\textbf{A.1}-\textbf{A.3} hold true to obtain the explicit evaluation for the OAEP. This asymptotic regime has been widely used in evaluating the performance of large-scale MIMO systems~\cite{hachem2008new,zhang2022asymptotic} and the strikingly simple expressions for the asymptotic performance have also been validated to be accurate even for the small-dimension systems. Before we proceed, we introduce the MID, which is important for deriving the bounds for the OAEP. Specifically, the MID of the concerned system can be given by 
\begin{equation}
\label{mid_exp}
I_{N,L,M}^{(n)}(\sigma^2)\overset{\bigtriangleup}{=}C(\sigma^2)+D(\sigma^2),
\end{equation}
where
\begin{equation}
\begin{aligned}
C(\sigma^2)&=\frac{1}{M}\log\det(\bold{I}_{N}+\frac{1}{\sigma^2}\BH\BH^{H}),
\\
D(\sigma^2)&=\frac{1}{Mn}\Tr[\BH\BH^{H}+\sigma^2\bold{I}_{N})^{-1}(\BH\BC^{(n)}+\sigma \BW^{(n)})
\\
& \times 
(\BH\BC^{(n)}+\sigma \BW^{(n)})^{H}  
-\BW^{(n)}(\BW^{(n)})^{H}].
\end{aligned}
\end{equation}
The following lemma indicates that the OAEP can be bounded. 
\begin{lemma} \cite[Eq. (77) and Eq. (89)]{hoydis2015second}
\label{bnd_err}
(Bounds for the OAEP) The upper and lower bounds of the OAEP are given by 
\begin{equation}
\mathbb{L}(r|\alpha, \beta, \tau)\le \Prob(r|\alpha, \beta, \tau ) \le \mathbb{U}(r|\alpha, \beta, \tau),
\end{equation}
where 
\begin{subequations}
\begin{align}
&\mathbb{U}(r|\alpha, \beta, \tau )=\lim_{\zeta  \downarrow 0} \limsup\limits_{ N \xrightarrow[]{\alpha, \beta,\tau} \infty} \nonumber
\\
&
 \Prob[\sqrt{Mn}(I^{(n)}_{N,L,M}(\sigma^2)-\overline{C}(\sigma^2))\le r+\zeta  ]\label{upp_bound},
\\
&\mathbb{L}(r|\alpha, \beta, \tau )=\inf\limits_{\{\Prob({\BC^{(n+1)})}\in \mathrm{P}(\mathcal{S}_{=}^{n+1})\}_{n=1}^{\infty}}
\lim\limits_{\zeta  \downarrow 0}
 \limsup\limits_{N \xrightarrow[]{\alpha, \beta,\tau} \infty} \nonumber
 \\
 &
 \Prob[\sqrt{Mn}(I^{(n+1)}_{N,L,M}(\sigma^2)-\overline{C}(\sigma^2))\le r-\zeta  ],
\end{align}
\end{subequations}
where $I^{(n)}_{N,L,M}$ is given in~(\ref{mid_exp}) and~(\ref{upp_bound}) is induced by the normalized Gaussian input $\BC^{(n)}\in \mathbb{C}^{M\times n}=\widetilde{\bold{C}}^{(n)}\left(\frac{1}{Mn}\Tr (\widetilde{\bold{C}}^{(n)}\widetilde{\bold{C}}^{(n),H})\right)^{-\frac{1}{2}}$, with $\widetilde{\bold{C}}^{(n)}\in \mathbb{C}^{M\times n}$ representing an i.i.d. Gaussian matrix.
\end{lemma}
\begin{remark} 
\label{ana_sphc}
Lemma~\ref{bnd_err} is important because: 1. it bounds the OAEP with the maximal energy constraint $\mathcal{S}$ by that with the equal energy constraint. 2. it turns the evaluation of the OAEP to the investigation of the MID. The switch from the maximal power constraint to the equal power constraint follows by introducing an auxiliary symbol~\cite[Lemma 39]{polyanskiy2010channel}. 
\end{remark} 

\subsection{Problem Formulation}
By Lemma~\ref{bnd_err}, the characterization of the OAEP is converted to the investigation of the MID distribution, which will be achieved by setting up a CLT for the MID. The error probability minimization problem can be formulated as
 \begin{equation}
 \label{P1_exp}
  \begin{aligned}
\mathcal{P}1:~&\min_{\bold{\Phi}}~\Prob_e(r|\alpha, \beta,\tau),~s.t.
\\
& \bold{\Phi}=\diag\left(\phi_1,\phi_2,...,\phi_L \right),
\\
&|\phi_l|=1, l=1,2,...L.
 \end{aligned}
 \end{equation}
In the following, we will set up a CLT for the MID and propose the phase shift design algorithm to solve $\mathcal{P}1$.


\section{Bounds for OAEP and Optimization}

In this section, we will characterize the asymptotic distribution of the MID and give the upper and lower bounds for the OAEP in closed forms. Then, we will propose a gradient-based algorithm to minimize the derived lower bound. Before presenting the main results, we introduce the following notations:
\begin{equation}
\bold{S}=\RT_{IRS}\bold{\Phi}\FR_{IRS}\bold{\Phi}^{H}\RT_{IRS},
\end{equation}
 \begin{equation}
 \label{basic_eq1}
 \begin{aligned}
 \begin{cases}
   &\delta(z)=\frac{1}{L}\Tr\bold{R}\BG(z),
   \\
 &\omega(z)=\frac{\delta(z)}{M}\Tr \FS\BF(z), 
 \\
 &\xi(z)=\frac{1}{M}\Tr\BD\BG(z),
\end{cases}
  \end{aligned}
 \end{equation}
where
 \begin{equation}
\bold{G}(z)=\left(z \bold{I}_{N}+ \overline{\omega}(z) \BD  + \frac{M \omega(z) \overline{\omega}(z)\bold{R}}{L\delta(z)}   \right)^{-1},
\end{equation}
\begin{equation}
\BF(z)=\left( \bold{I}_{L}+\overline{\omega}(z)\delta(z) \FS   \right)^{-1},
\end{equation}
and $\overline{\omega}=\left( 1+\omega(z)+\xi(z) \right)^{-1}$. The resolvent matrix of the gram matrix $\BH\BH^{H}$ is defined as
\begin{equation}
\label{resol_def}
\BQ(z)=\left(z\bold{I}_{N}+\BH\BH^{H}\right)^{-1}.
\end{equation}
For ease of presentation, we summarize some important notations in Table~\ref{var_list} and omit $(z)$ for brevity.

 \begin{table}[!htbp]
\centering
\caption{Useful notations.}
\label{var_list}
\begin{tabular}{|cc|cc|}
\toprule
Notation& Value &  Notation & Value \\
\midrule
$\gamma_{R}$ & $\frac{1}{L}\Tr\FR\BG\FR\BG$
&
$\gamma_{R,I}$ & $\frac{1}{L}\Tr\FR\BG^2$
 \\
  $\gamma_{S}$ & $\frac{1}{M \delta^2}\Tr\FS^2\BF^2$
  &
    $\gamma_{S,I}$ & $\frac{1}{M \delta^2}\Tr\FS\BF^2$
    \\
  $\gamma_{D}$  & $\frac{1}{M}\Tr\BD\BG\BD\BG$ 
  & $\gamma_{R,D}$ &  $\frac{1}{L}\Tr  \FR\BG\BD\BG  $
  \\
    $\gamma_{D,I}$  & $\frac{1}{M}\Tr\BD\BG^2$ 
& $\gamma_{I}$ & $\frac{1}{L}\Tr\BG^2$
\\
    $\Delta_{S}$  & $1-\gamma_S\overline{\omega}^2$ 
  \\
  \bottomrule
\end{tabular}
\end{table}
Noticing that $\E I_{N,L,M}^{(n)}(\sigma^2)=\E C(z)$ (since $\E D(z)=0$), the approximation of the MID can be obtained by the following lemma.
 \begin{lemma} 
 \label{lemma_mean}
 Given assumptions~\textbf{A.1} to \textbf{A.3}, $C(z)$ can be approximated by
\begin{equation}
\label{mean_appro}
C(z) =  \overline{C}(z)+\BO(M^{-1}),
\end{equation}
where $0<z<\infty$ and
  \begin{equation}
\begin{aligned}
&  \overline{C}(z)=\frac{1}{M} (N\log(z^{-1})+ \log\det(\BG^{-1})+
  \\
  &
\log\det(\BF^{-1})
  +M\log(1+\xi+\omega)-2M\omega\overline{\omega} - M\xi\overline{\omega} ).
\end{aligned}
  \end{equation}
  There also holds true that, for any $\BM$ with bounded norm, i.e., $\limsup_{N\ge 1}\| \BM\|<\infty$,
  \begin{equation}
  \label{trace_appro}
 \begin{aligned}
&  \frac{\E\Tr\BM\BQ}{L} \!=\! \frac{\Tr\BM\BG}{L}+\BO(M^{-2}),~\frac{\E\Tr\FR\BQ}{L}\!=\!\delta+\BO(M^{-2}),
  \\
 &\frac{\E\Tr\BZ\BZ^{H}\BQ}{M}\!=\!\omega+\BO(M^{-2}),~\frac{\E\Tr\BD\BQ}{M}\!=\!\xi+\BO(M^{-2}).
  \end{aligned}
  \end{equation}
\end{lemma}
\begin{remark} The approximation in~(\ref{mean_appro}) coincides with that in~\cite[Proposition 1]{moustakas2023reconfigurable} and the convergence rate $\BO(M^{-1})$ can be obtained by the Gaussian tools in~\cite{zhang2022asymptotic}. Furthermore, $\BG$ is a good approximation for $\E\BQ$ in trace operations.
 \end{remark}
By Lemma~\ref{lemma_mean}, we have obtained the evaluation for $\E I_{N,M,L}(z)=\overline{C}(z)+\BO(M^{-2})$. Next, we will investigate the asymptotic distribution of $I_{N,M,L}(z)$.
 \begin{figure*}
 \vspace{-0.7cm}
\begin{equation}
\label{def_pis_qs}
 \begin{aligned}
\bold{\Pi}&=
\begin{bmatrix}
1-\frac{M\gamma_{S}\overline{\omega}^2\gamma_{R}}{L\delta^2} & -\frac{M\overline{\omega}^2 \gamma_{R}\gamma_{S,I} }{L\delta^2\Delta_S}- \frac{\overline{\omega}^2\gamma_{R,D}}{\Delta_S}
\\
  -\frac{\gamma_S\overline{\omega}^2(\frac{M\gamma_{S,I}\gamma_R}{L\delta^2}+\gamma_{R,D})}{\delta^2} &1-\frac{\overline{\omega}^2 (\frac{M\gamma_{S,I}\gamma_{R}}{L\delta^2}+ \gamma_{R,D})\gamma_{S,I} }{\delta^2\Delta_S}- \frac{\overline{\omega}^2(\frac{\gamma_{S,I}\gamma_{R,D}}{\delta^2}+\gamma_{D})}{\Delta_S}
\end{bmatrix},~~
  \bold{q}_{I}= \begin{bmatrix}
\gamma_{R,I}
&
\frac{\gamma_{S,I}\gamma_{R,I}}{\delta^2}+\gamma_{D,I}
 \end{bmatrix}^{T},
\\
    \bold{p}_{I}&= \begin{bmatrix}
\frac{M\overline{\omega}^2 \gamma_{S} \gamma_{R,I} }{L\delta^2} &
\frac{M\overline{\omega}^2 \gamma_{R,I}\gamma_{S,I} }{L\delta^2\Delta_S}+ \frac{M\overline{\omega}^2\gamma_{D,I}}{L\Delta_S}
 \end{bmatrix}^{T},~~
     \bold{q}_{RD}= \begin{bmatrix}
\frac{M\gamma_{S,I}\gamma_R}{L\delta^2} + \gamma_{R,D}
&
\frac{M\gamma_{S,I}^2\gamma_{R}}{L\delta^4}+
\frac{2\gamma_{S,I}\gamma_{R,D}}{\delta^2}+\gamma_{D}
 \end{bmatrix}^{T}
 \end{aligned}
 \end{equation}
 \hrulefill
 \vspace{-0.3cm}
\end{figure*}

\subsection{CLT for the MID}
With the equal energy constraint in~(\ref{sph_cons}), the asymptotic distribution of the MID is given in the following theorem. 
\begin{theorem}
\label{clt_the}
 (CLT for the MID)
Given assumptions~\textbf{A.1} to~\textbf{A.3} and a sequence of $\BC^{(n)}$, the distribution of the MID converges to a Gaussian distribution, i.e.,
\begin{equation}
\sqrt{\frac{{Mn}}{V_n}}(I_{N,L,M}^{(n)}(\sigma^2)-\overline{C}(\sigma^2))  \xrightarrow[{N  \xrightarrow[]{\alpha, \beta,\tau}\infty}]{\mathcal{D}}  \mathcal{N}(0,1),
\end{equation}
where $V_n$ is given by
\begin{equation}
\begin{aligned}
\label{clt_var}
V_n&=-\tau\log(\Xi)+\alpha-\frac{\sigma^4 (\gamma_{I}+\bold{p}_{I}^{T}\bold{\Pi}^{-1}\bold{q}_{I})}{\beta}
\\
&
+\frac{\tau \overline{\omega}^4\Tr\BA_{n}^2}{M}(\frac{ [\bold{\Pi}^{-1}]_{(2)}\bold{q}_{RD}  }{\Delta_S^2}+\frac{\gamma_S}{\Delta_S}),
\end{aligned}
\end{equation}

with 
\begin{equation}
\label{def_A}
 \BA_{n}=\bold{I}_M-\frac{1}{n}\BC^{(n)}(\BC^{(n)})^H,
 \end{equation}
 \begin{equation}
\Xi=\det(\bold{\Pi})\Delta_{S}.
\end{equation}
The parameters $\bold{\Pi}$, $\bold{p}_{I}$, $\bold{q}_{I}$, and $\bold{p}_{RD}$ are given in~(\ref{def_pis_qs}) at the top of the next page and $\gamma_{(\cdot)}$ is given in Table~\ref{var_list}. $[\bold{\Pi}^{-1}]_{(2)}$ denotes the second row of $[\bold{\Pi}^{-1}]$.
\end{theorem}
\begin{proof}
The proof of Theorem~\ref{clt_the} is sketched in Appendix~\ref{proof_clt_the}.
\end{proof}
\begin{remark} When $\FS=\bold{0}$ and $\BD=\bold{I}_{N}$, Theorem~\ref{clt_the} degenerates to the result for single-hop Rayleigh MIMO channels~\cite[Theorem 2]{hoydis2015second}. When there is no channel correlation and direct link, i.e., $\FR=\bold{I}_{N}$ and $\FR_{IRS}=\FT_{IRS}=\bold{I}_{L}$ and $\BD=\bold{0}$, Theorem~\ref{clt_the} degenerates to the result for Rayleigh-product MIMO channels~\cite[Theorem 2]{zhang2022second}.
\end{remark}

Theorem~\ref{clt_the} indicates that the asymptotic distribution of the MID is a Gaussian distribution whose mean and variance are determined by the system of equations in~(\ref{basic_eq1}) and variables in Table~\ref{var_list}. By Lemma~\ref{bnd_err}, we can obtain the upper and lower bounds for the OAEP.

\subsection{Upper and Lower Bounds for the OAEP}
\begin{theorem} \label{the_oaep}
(Bounds for the OAEP)
 Given the second-order coding rate $r$, the OAEP of the IRS-aided MIMO system, $\Prob_e(r|\tau,\alpha, \beta)$ , is bounded by
\begin{equation}
\label{low_bnd_exp}
\Prob_e(r|\alpha, \beta,\tau) \ge 
\begin{cases}
\Phi(\frac{r}{\sqrt{V_{-}}}),~~ r\le 0,
\\
\frac{1}{2},~~ r> 0,
\end{cases}
\end{equation}
\begin{equation}
\Prob_e(r|\alpha, \beta,\tau)  \le \Phi(\frac{r}{\sqrt{V_{+}}}),
\end{equation}
where
\begin{equation}
\begin{aligned}
\label{var_upp_low}
V_{-}&=-\tau\log(\Xi)+\alpha-\frac{\sigma^4 (\gamma_{I}+\bold{p}_{I}^{T}\bold{\Pi}^{-1}\bold{q}_{I})}{\beta},
\\
V_{+}&=V_{-}+\overline{\omega}^4(\frac{ [\bold{\Pi}^{-1}]_{(2)}\bold{q}_{RD}  }{\Delta_S^2}+\frac{\gamma_S}{\Delta_S}).
\end{aligned}
\end{equation}
\end{theorem}
\begin{proof} The proof can be obtained by similar steps in~\cite[Theorem 3]{zhang2022second} and is omitted here.
\end{proof}


\subsection{Lower Bound Minimization}
Without the closed-form expression for the packet error rate, $\mathcal{P}1$ is difficult to optimize. Thus, we first reformulate $\mathcal{P}1$ to minimize the lower bound of the OAEP in~(\ref{low_bnd_exp}) as
 \begin{equation}
  \begin{aligned}
\mathcal{P}2:~&K(\bold{\Phi})=\min_{\bold{\Phi}}~ \Phi(\frac{r}{\sqrt{V_{-}}}),~s.t.
\\
& \bold{\Phi}=\diag\left(\phi_1,\phi_2,...,\phi_L \right),
\\
&|\phi_l|=1, l=1,2,...L.
 \end{aligned}
 \end{equation}
 In order to tackle the non-convexity of the phase shift constraint, we use the Armijo-Goldstein (AG) line search method to find a decrease of the objective function in the negative gradient direction. The detailed gradient descent algorithm is given in Algorithm~\ref{gra_alg}. Similar to~\cite{zhang2022outage}, the gradient can be computed by the chain rule and the derivatives of $\delta$, $\omega$, $\xi$ can be obtained by taking derivative on both sides of~(\ref{basic_eq1}), which are omitted here.
\begin{algorithm} 
\caption{Lower Bound Minimization Algorithm for $\mathcal{P}$2} 
\label{gra_alg} 
\begin{algorithmic}[1] 
\REQUIRE  $\bm{\theta}^{\left(0 \right)}$, scaling factor $0<c<1$, control parameter $0<\kappa<1$, and initial searching stepsize $\lambda_{0}$.
Set $s=0$. 
\REPEAT
\STATE Compute the gradient
$$\nabla_{\bm{\theta}}K(\bold{\Phi})=  (\frac{\partial K(\bold{\Phi})}{\partial \theta_{1}}, \frac{\partial K(\bold{\Phi})}{\partial \theta_{2}},..., 
\frac{\partial K(\bold{\Phi})}{\partial \theta_{L}})^{T},$$ 
and its direction $\bold{g}^{(s)}=\frac{\nabla_{\bm{\theta}}K(\bold{\Phi})}{\|\nabla_{\bm{\theta}}K(\bold{\Phi}) \|}$.
\STATE $\lambda \leftarrow\lambda_{0}$.
\WHILE{$K(\bold{\Phi}^{(s)})-K( \diag[\exp(\jmath\bm{\theta}^{(s)}-\lambda \jmath \bold{g}^{(s)}) ])<\lambda\kappa \| \nabla_{\bm{\theta}}K(\bold{\Phi}^{(s)})\|$} 
\STATE  $\lambda \leftarrow c\lambda$.
\ENDWHILE 
	\STATE $\bm{\theta}^{(s+1)} \leftarrow \bm{\bold{\theta}}^{(s)}-\lambda \bold{g}^{(s)}$.
	\STATE $\bold{\Phi}^{(s+1)} \leftarrow \diag[\exp(\jmath\bm{\theta}^{(s+1)})]$.
\STATE $s \leftarrow  s+1$.
\UNTIL Convergence.
\ENSURE  $\bm{\theta}^{(s)}, \bold{\Phi}^{(s)}$.
\end{algorithmic}
\end{algorithm}


\section{Numerical Simulations}
\label{sec_simu}
In this section, we validate the theoretical results by numerical simulations.
\subsection{Simulation Settings}
We consider an IRS-aided MIMO system, where the distances of the BS-user, BS-IRS, and IRS-user channels are set to be $d_{BS-U}=50$ m, $d_{BS-IRS}=40$ m, and $d_{IRS-U}=50$ m, respectively. The path loss of the three channels are given by
\begin{equation}
\begin{aligned}
\rho_{BS-IRS}&=\frac{C_{1}}{d_{BS-IRS}^{\alpha_1}},~\rho_{IRS-U}=\frac{C_{2}}{d_{IRS-U}^{\alpha_1}}, 
\\
\rho_{BS-U}&=\frac{C_3}{d_{IRS-U}^{\alpha_2}}.
\end{aligned}
\end{equation}
Here $\alpha_i,i=1,2$ represents the path loss exponent and $C_i,i=1,2,3$ represents the $1$ meter reference path loss. The parameters are set as $\alpha_1=2.2$, $\alpha_2=3.67$ , $C_1=10^{-2.305}$, $C_2=0.05$, $C_3=10^{-2.595}$~\cite{kammoun2020asymptotic}, and $\sigma^2=-80$ dBm. The correlation matrix for the uniform linear array is constructed by~\cite{yong2005three,zhang2022secrecy},
\begin{equation}
\begin{aligned}
&[\BC(d_r,\eta,\delta,G)]_{m,n}
\\
&
=\int_{-180}^{180}\frac{1}{\sqrt{2\pi\delta^2}}e^{\jmath \frac{2\pi}{\lambda} d_r(m-n)\sin(\frac{\pi\phi}{180})-\frac{(\phi-\eta)^2}{2\delta^2} }\mathrm{d}\phi,
\end{aligned}
\end{equation}
where $d_r$ represents the relative antenna spacing (in wavelengths), which is set to be $0.5$, $m$ and $n$ denotes the indices of antennas, $\eta$ and $\delta^2$ represent the mean angle and the mean-square angle spreads in degree, respectively, and $G$ is the number of antennas.

In Fig.~\ref{simu_clt},  the Gaussianity is validated by comparing the empirical probability density function (PDF) of the normalized MID and the PDF of the standard Gaussian distribution. The parameters are set as $\FR=[\BC(0.5,10,10,16)]$, $\BD=[\BC(0.5,5,5,16)]$, $\FT_{IRS}=[\BC(0.5,0,5,32)]$, $\FR_{IRS}=[\BC(0.5,15,15,32)]$, $n=36$ and the transmit power $P=15$ dbm W. The number of the Monte-Carlo realizations is $10^{6}$ and $\BC^{(n)}=\bold{1}_{M}\bold{1}_{n}^{T}$. It can be observed that the histogram of the normalized MID fits the PDF of the standard Gaussian distribution well, which also validates that the closed-form expressions for the mean and variance in Theorem~\ref{clt_the} are accurate.
\begin{figure}[t!]
\centering\includegraphics[width=0.45\textwidth]{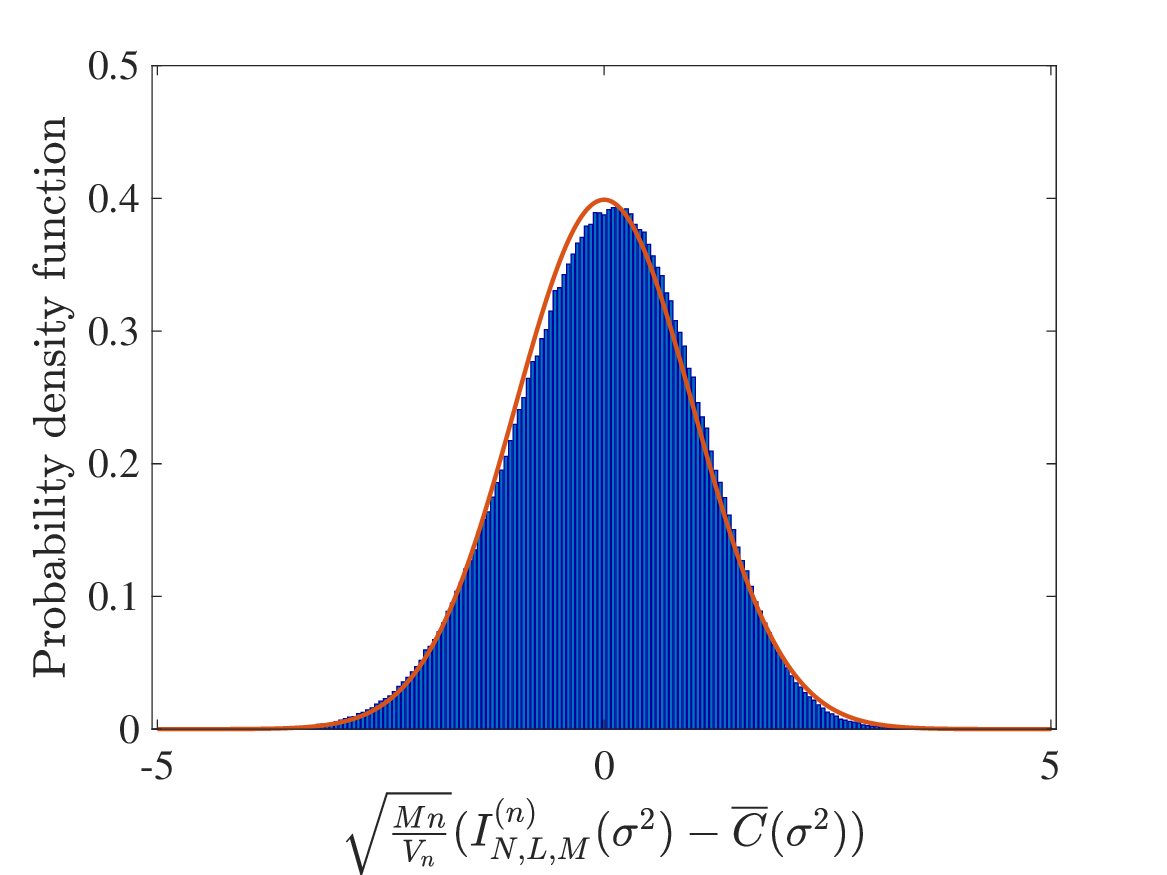}
\caption{Fitness of the CLT.}
\label{simu_clt}
\end{figure}

With the same settings as Fig.~\ref{simu_clt}, Fig.~\ref{bounds_fig} depicts the upper and lower bounds of the OAEP in Theorem~\ref{the_oaep} for different second-order coding rates $r$, which is chosen to be proportional to $\overline{C}(\sigma^2)$. The corresponding coding rate is $R=\overline{C}(\sigma^2)+\frac{r}{\sqrt{nM}}$. It can be observed that the gap between the upper and lower bounds is small, which indicates the tightness of the derived bounds.

\begin{figure}[t!]
\centering\includegraphics[width=0.45\textwidth]{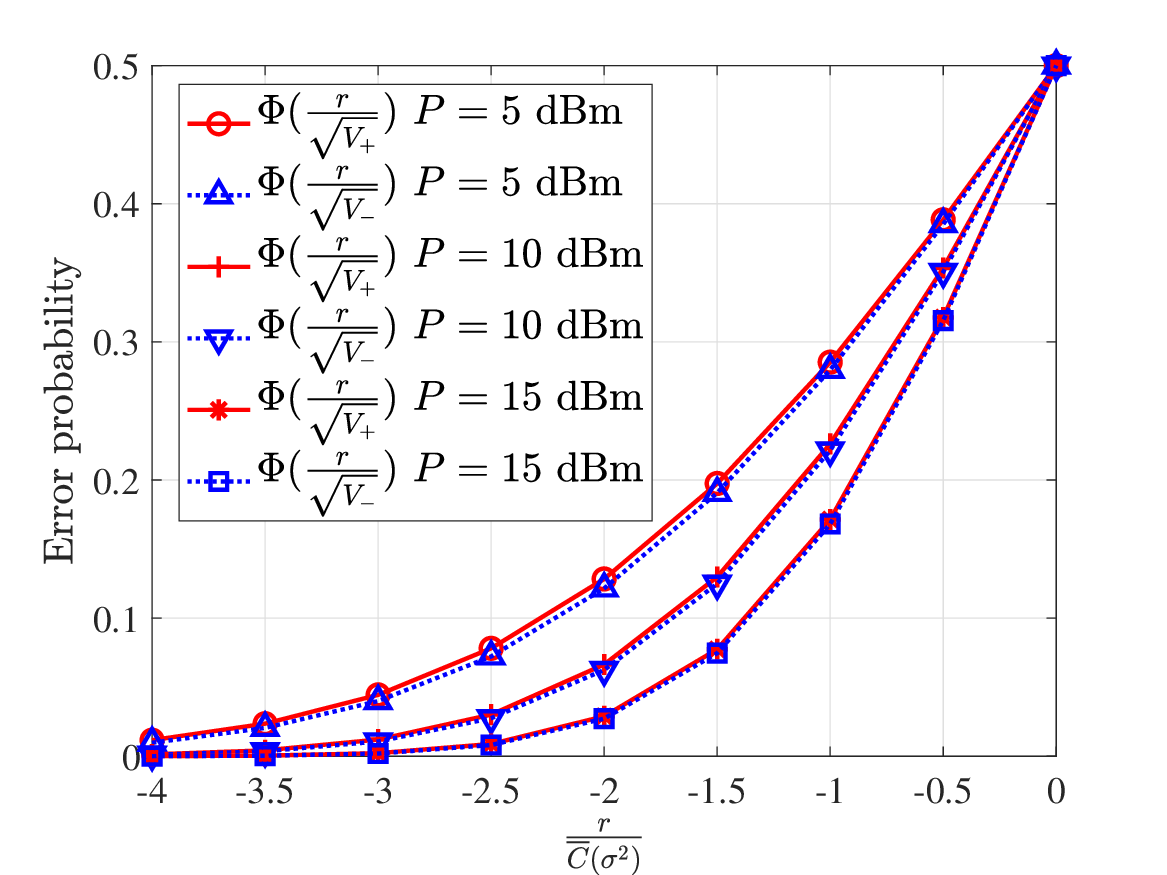}
\caption{Upper and lower bounds of error probability.}
\label{bounds_fig}
\end{figure}

In Fig~\ref{LDPC_performance}, we compare the lower bound of the OAEP and the LDPC error probability with and without phase shift optimization. Here we set $\FR=\BD=\BI_{M}$ and the $1 \slash 2$ LDPC code  in WiMAX standard is adopted~\cite{8303870}. We generate the inputs for the LDPC encoder uniformly and adopt a bit interleaved coded modulation (BICM) with a random interleaver before the QPSK modulation. Under such circumstances, the per antenna coding rate in nats is $R=\log(2)$. The ML demodulator is utilized for the receive signal~\cite{mckay2005capacity} with
\begin{equation}
L(s_{i}|\bold{r},\BH)=\log\frac{\sum_{\bold{c}\in \mathcal{C}^{(i)}_{1} } p(\bold{r}|\bold{c},\BH) }{\sum_{\bold{c}\in \mathcal{C}^{(i)}_{0} } p(\bold{r}|\bold{c},\BH) },
\end{equation}
where $L(s_{i}|\bold{r},\BH)$ represents the log likelihood ratio of $i$-th bit $s_{i}\in \{0,1\}$ and $p(\bold{r}|\bold{c},\BH)$ is the conditional PDF of the receive signal $\bold{r}$. Here $\mathcal{C}^{(i)}_{1}=\{\bold{c}| c_{i}=1,\bold{c}\in \mathcal{C}   \}$ and $\mathcal{C}^{(i)}_{0}$ denote the set of the codewords whose $i$-th digit are $1$ and $0$, respectively. The output of the demodulator is then decoded by the soft-decision LDPC decoder. The length of the LDPC codes is $l=576$ bits and the corresponding blocklength is $n=\frac{l}{2M}=36$. The second-order coding rate is $r=\frac{R-\overline{C}(\sigma^2)}{\sqrt{Mn}}$. It can be observed from Fig.~\ref{LDPC_performance} that Algorithm~\ref{gra_alg} could effectively decrease the lower bound of the OAEP. Furthermore, the lower bound of the OAEP and the performance of the LDPC code have a similar slope. Finally, the phase shifts obtained by Algorithm~\ref{gra_alg} could also effectively decrease the error rate of the LDPC code.
\begin{figure}[t!]
\centering\includegraphics[width=0.45\textwidth]{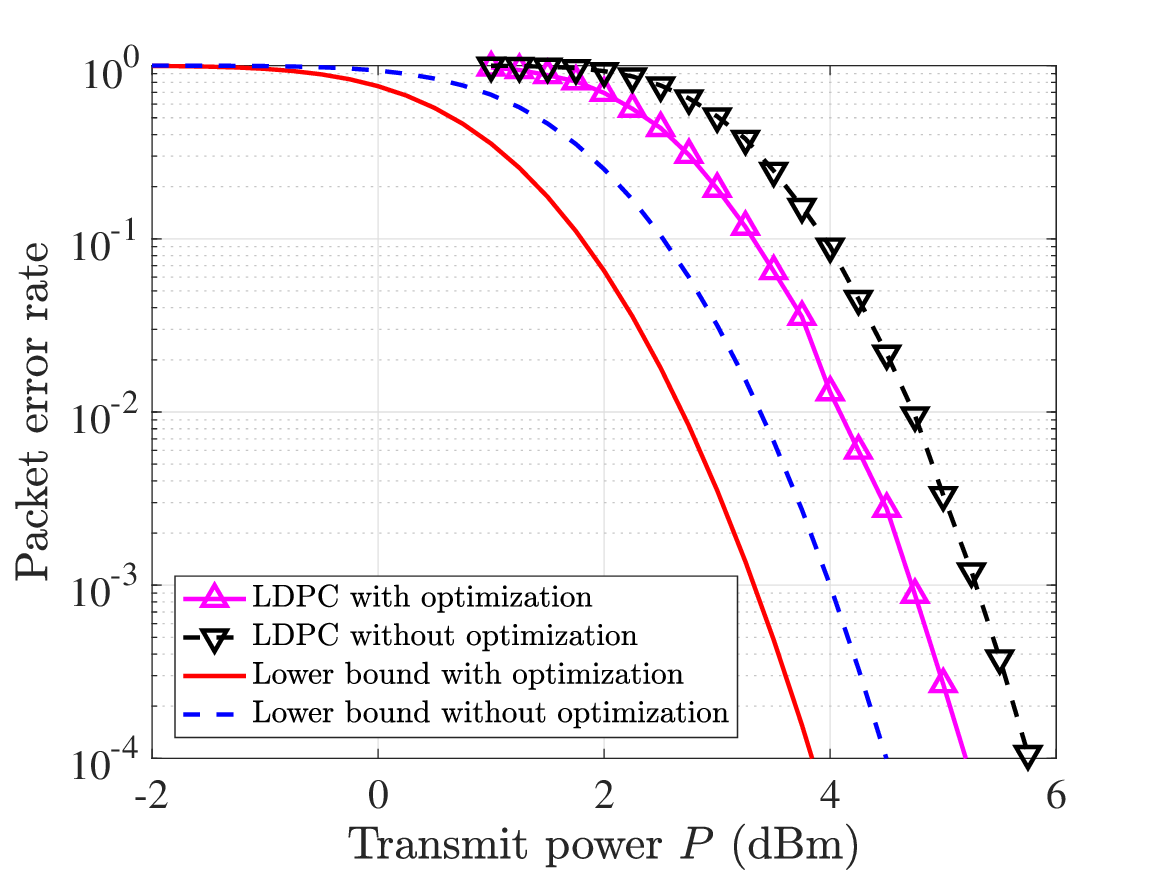}
\caption{LDPC performance.}
\label{LDPC_performance}
\end{figure}

%

\section{Conclusion}
\label{sec:conclusion}
In this work, the OAEP of IRS-aided MIMO systems was investigated by RMT when the coding rate is within a $\BO(\frac{1}{\sqrt{Mn}})$ perturbation of the asymptotic capacity. To this end, we first set up a CLT for the MID and gave the closed-form expressions for the mean and variance in the asymptotic regime, which take previous results for the single-hop MIMO channel and Rayleigh-product channel as special cases. The CLT was then utilized to obtain the upper and lower bounds of the OAEP. Finally, a gradient based algorithm was proposed to minimize the proposed lower bound by optimizing the phase shift. Numerical simulations validate the fitness of the CLT and the effectiveness of the proposed algorithm. It was also observed that the error rate of the LDPC code and the lower bound of the OAEP have a similar slope and the proposed algorithm is also effective in optimizing the packet error rate of the LDPC code.




\appendices
\section{Proof of Theorem~\ref{clt_the}}
\label{proof_clt_the}
In this appendix, we set up the CLT by showing that the characteristic function of the MID converges to that of the Gaussian distribution. This approach has been used in the second-order analysis of the MIMO channels~\cite{hachem2008new,zhang2022asymptotic}. We also derive the closed-form mean and variance during the evaluation of the characteristic function. The characteristic function of the MID is presented as 
\begin{equation}
\Psi^{\BW,\BY,\BX,\BU}_{n}(u)=\E\Phi^{\BW,\BY,\BX,\BU}_{n} ,
\end{equation}
where $\Phi^{\BW,\BY,\BX,\BU}(u)= e^{\jmath u \gamma_{n}^{\BW,\BY,\BX,\BU}}$ and $\gamma_{n}^{\BW,\BY,\BX,\BU}= \sqrt{nM}I_{N,L,M}^{(n)}(\sigma^2)$. There are four random matrices $\BW,\BY,\BX,\BU$ included in the characteristic function. The evaluation of $\Psi^{\BW,\BY,\BX,\BU}_{n}(u)$ is difficult due to the exponential structure. To overcome the difficulty, we first investigate its derivative with respect to $u$, i.e.,
\begin{equation}
\label{der_psi}
\frac{\partial \Psi^{\BW,\BY,\BX,\BU}_{n}(u)}{\partial u}=\E\jmath\gamma_{n}^{\BW,\BY,\BX,\BU} e^{\jmath u\gamma_{n}^{\BW,\BY,\BX,\BU}}.
\end{equation}
Then, by taking integral over the RHS of~(\ref{der_psi}) with respect to $u$, we could find $V_n$ and $\overline{\gamma}_{n}$ such that 
\begin{equation}
\label{conv_cha}
\E e^{\jmath u \frac{\gamma_{n}^{\BW,\BY,\BX,\BU}-\overline{\gamma}_{n}}{\sqrt{V_n}}} \xrightarrow[]{N  \xrightarrow[]{\alpha,\beta, \tau}\infty} e^{-\frac{u^2}{2}},
\end{equation}
which concludes the asymptotic Gaussianity of the MID. The expectation over the four random matrices $\BW$, $\BY$, $\BX$ $\BU$ will be handled iteratively. Specifically, in the first step, we will take the expectation over $\BW$. Then in the second step, we will take the expectation over $\BY$, $\BX$, and $\BU$. Finally, we derive $V_n$ and $\overline{\gamma}_n$ to conclude~(\ref{conv_cha}). Therefore, the proof includes three steps.

\subsection{Step:1 Expectation over $\BW$}
In this step, we will provide an approximation for $\frac{\partial \Psi^{ \BW,\BY,\BX,\BU}(u)}{\partial u}$, which only relies on $\BY$, $\BX$, $\BU$ by taking the expectation over $\BW$. Given assumptions~\textbf{A.1}-\textbf{A.3}, by using the Gaussian tools (the integration by parts formula~\cite[Eq. (40)]{zhang2022asymptotic} and Nash-Poincar{\'e} inequality~\cite[Eq. (18)]{hachem2008new}) and following Step 1 in~\cite[Appendix D.D, Eq. (221) to (240)]{hoydis2015second}, the characteristic function $ \Psi^{ \BW,\BY,\BX,\BU}(u)$ can be approximated by 
\begin{equation}
\label{fir_eva_W}
\begin{aligned}
& \frac{\partial  \Psi^{\BW,\BY,\BX,\BU}(u)}{\partial u}= \jmath u \E\mu^{\BY,\BX,\BU}_{n}\Phi_{n}^{\BY,\BX,\BU}
\\
&
 -\E\frac{u^2}{2}\nu^{\BY,\BX,\BU}_{n} \Phi_{n}^{\BY,\BX,\BU} +\BO(\frac{1}{M})
\\
&
=C_1+C_2+\BO(\frac{1}{M}),
\end{aligned}
\end{equation}
where
\begin{align*}
&\Phi_{n}^{\BY,\BX,\BU}=e^{ \iota^{\BY,\BX,\BU}_{n}},
\\
&
\iota^{\BY,\BX,\BU}_{n}=\jmath u \mu^{\BY,\BX,\BU}_{n}-\frac{u^2}{2}\nu^{\BY,\BX,\BU}_{n}
+\frac{\jmath u^3 \eta^{\BY,\BX,\BU}_{n}}{3},
\\
&\mu_{n}^{\BY,\BX,\BU}\!=\!\sqrt{\frac{n}{M}}\log\det\left(\bold{I}_N+\BH\BH^{H}\right)\! - \! \frac{n}{\sqrt{Mn}}\Tr\BQ\BH\BA\BH^{H},
\\
&\nu_n^{\BY,\BX,\BU}=\frac{n}{Mn}\Tr(\BQ\BH\BH^{H})^2+\frac{2\sigma^2n}{Mn}\Tr\BQ^2\BH\frac{\BX\BX^{H}}{n}\BH^{H},
\\
&\eta_n^{\BY,\BX,\BU}=\frac{n}{\sqrt{n^3M^3}}\Tr(\BQ\BH\BH^{H})^3
\\
&
+\frac{3\sigma^2n}{\sqrt{n^3M^3}}\Tr\BQ^2\BH\BH^{H}\BQ\BH\frac{\BX\BX^{H}}{n}\BH^{H}. \numberthis
\end{align*}
By~(\ref{fir_eva_W}), the evaluation of the characteristic function resorts to those of $C_1$ and $C_2$, which will be given in the next step.

\subsection{Step:2 Expectation over $\BY$, $\BX$, and $\BU$}
By $\log(\det(\sigma^2\bold{I}_N)+\BH\BH^{H})= \int_{\sigma^2}^{\infty}\frac{N}{z}- \Tr \BQ(z)  \mathrm{d}z$~\cite[Eq (4)]{zhang2021bias}, $C_1$ can be evaluated as
\begin{equation}
\label{C_1_eva}
\begin{aligned}
&
C_1=(\jmath\sqrt{\frac{n}{M}}\int_{\sigma^2}^{\infty} \frac{N\E\Phi_{n}^{\BY,\BX,\BU}}{z}
-\E\Tr\BQ(z)\Phi_{n}^{\BY,\BX,\BU}\mathrm{d}z)
\\
&
-\frac{n}{\sqrt{Mn}}\E\Tr\BQ\BH\BA\BH^{H}\Phi_{n}^{\BY,\BX,\BU}+o(1)
\\
&=C_{1,1}+C_{1,2}+o(1).
\end{aligned}
\end{equation}
Here $C_{1,1}$ can be handled through similar analysis as~\cite[Section VI]{zhang2022asymptotic} to obtain
\begin{equation}
\begin{aligned}
C_{1,1}&=
\jmath \sqrt{nM}\overline{C}(\sigma^2)\E\Phi_{n}^{\BY,\BX,\BU}
\\
&-\frac{un}{M}(-\log(\Xi))\E\Phi_{n}^{\BY,\BX,\BU}+o(1).
\end{aligned}
\end{equation}
Next, we focus on the evaluation of $C_{1,2}$. Denote $\MS_{+}=\RT_{IRS}\bold{\Phi}\LR_{IRS}$ and $\BZ=\LR\BX\MS_{+}$. By the integration by parts formula~\cite[Eq. (40)]{zhang2022asymptotic}, we can evaluate $\E\Tr\BQ\BH\BA\BH^{H}\Phi_{n}^{\BY,\BX,\BU}$ in~(\ref{U12app}) at the top of the next page.
\begin{figure*}
\vspace{-0.7cm}
\begin{align*}
\label{U12app}
&\E\Tr\BQ\BH\BA\BH^{H}\Phi_{n}^{\BY,\BX,\BU}
=\sum_{i,j}\E[\BY]_{j,i}^{*}[\BZ^{H}\BQ\BH\BA]_{j,i}\Phi_{n}^{\BY,\BX,\BU}+\E[\BU]_{j,i}^{*}[\LD\BQ\BH\BA]_{j,i}\Phi_{n}^{\BY,\BX,\BU}
=\frac{1}{M}\E\Tr\BQ\BZ\BZ^{H}\Tr\BA
\\
&
-\frac{1}{M}\E\Tr\BQ\BZ\BZ^{H}\E\Tr\BQ\BH\BA\BH^{H}\E\Phi_{n}^{\BY,\BX,\BU}
+\E[\frac{\jmath u \sqrt{n}}{M^{\frac{3}{2}}}\Tr\BQ\BZ\BZ^{H}\BQ\BH\BA\BH^{H}
-\frac{\jmath u n}{\sqrt{Mn}M}\Tr\BQ\BZ\BZ^{H}\BQ\BH\BA\BH^{H}\BQ\BH\BA\BH^{H}
\\
&
+\frac{\jmath u n}{\sqrt{Mn}}\frac{\Tr\BZ\BZ^{H}\BQ\BH\BA^2\BH^{H}\BQ}{M}]\E\Phi_{n}^{\BY,\BX,\BU}
+\E\frac{\Tr\BD\BQ}{M}\Tr\BA\Phi_{n}^{\BY,\BX,\BU}-\frac{1}{M}\E\Tr\BD\BQ\E\Tr\BQ\BH\BA\BH^{H}\Phi_{n}^{\BY,\BX,\BU}  \numberthis
\\
&
+\E [\frac{\jmath u \sqrt{n}}{M^{\frac{3}{2}}}\Tr\BQ\BD\BQ\BH\BA\BH^{H}
-\frac{\jmath u n}{\sqrt{Mn}M}\Tr\BQ\BD\BQ\BH\BA\BH^{H}\BQ\BH\BA\BH^{H}
+\frac{\jmath u n}{\sqrt{Mn}}\frac{\Tr\BD\BQ\BH\BA^2\BH^{H}\BQ}{M}
]\E\Phi_{n}^{\BY,\BX,\BU}+o(1)
\end{align*}
\hrulefill
\end{figure*}
By noticing $\Tr\BA=0$, moving $-\frac{1}{M}\E\Tr\BQ\BZ\BZ^{H}\E\Tr\BQ\BH\BA\BH^{H}\E\Phi_{n}^{\BY,\BX,\BU}$ and $-\frac{1}{M}\E\Tr\BQ\BD\E\Tr\BQ\BH\BA\BH^{H}\E\Phi_{n}^{\BY,\BX,\BU}$ to the LHS of~(\ref{C_2_eva}) and using the approximations in~(\ref{trace_appro}), we have
\begin{equation}
\label{QHAHPhiYZ}
\begin{aligned}
&
\E\Tr\BQ\BH\BA\BH^{H}\Phi_{n}^{\BY,\BX,\BU}
\\
&
=
\frac{\overline{\omega}n}{\sqrt{Mn}}\frac{\E\Tr\BQ\BZ\BZ^{H}\BQ\BH\BA\BH^{H}\BQ\BH\BA\BH^{H}}{M}\E\Phi_{n}^{\BY,\BX,\BU}
\\
&
-\frac{\overline{\omega}n}{\sqrt{Mn}}\frac{\E\Tr\BZ\BZ^{H}\BQ\BH\BA^2\BH^{H}\BQ}{M}\E\Phi_{n}^{\BY,\BX,\BU}
\\
&
\frac{\overline{\omega}n}{\sqrt{Mn}}\frac{\E\Tr\BQ\BD\BQ\BH\BA\BH^{H}\BQ\BH\BA\BH^{H}}{M}\E\Phi_{n}^{\BY,\BX,\BU}
\\
&
-\frac{\overline{\omega}n}{\sqrt{Mn}}\frac{\E\Tr\BD\BQ\BH\BA^2\BH^{H}\BQ}{M}\E\Phi_{n}^{\BY,\BX,\BU}+o(1)
\\
&
=\frac{\overline{\omega}^2 n}{\sqrt{Mn}}(\frac{\E\Tr\BQ\BZ\BZ^{H}\BQ\BH\BA^2\BH^{H} }{M}
\\
&
+\frac{\E\Tr\BQ\BD\BQ\BH\BA^2\BH^{H} }{M})\E\Phi_{n}^{\BY,\BX,\BU}+o(1)
\\
&= \frac{\overline{\omega}^4 n}{\sqrt{Mn}}  \frac{\Tr\BA^2}{M}(\chi+\zeta_D+\zeta_D+\Upsilon_D)+o(1). 
\end{aligned}
\end{equation}
where $\chi=\frac{\E\Tr\BQ\BZ\BZ^{H}\BQ\BZ\BZ^{H}}{M}$, $\zeta_D=\frac{\E\Tr\BQ\BZ\BZ^{H}\BQ\BD}{M}$, and $\Upsilon_D=\frac{\E\Tr\BQ\BD\BQ\BD}{M}$. By similar computation in~(\ref{U12app}),  we can evaluate $\chi+\zeta_D$ and $\zeta_D+\Upsilon_D$ by
\begin{align}
\chi+\zeta_D &=\frac{1}{\Delta_S}(\zeta_{D}+\gamma_{S}+\frac{M\gamma_{S,I}\zeta_R}{L\delta^2}),
\\
\zeta_{D}+\Upsilon_D &=\frac{\gamma_{S,I}\E \Tr\BQ\BD\BQ\FR }{\delta^2 L\Delta_S}+\frac{\Upsilon_D}{\Delta_S}+o(1),\nonumber
\end{align}
where $\zeta_{R}=\frac{\E\Tr\BZ\BZ^{H}\BQ\FR\BQ}{M}$. Then, we can obtain
\begin{equation}
 \begin{aligned}
& \chi+2\zeta_D+\Upsilon_D=\frac{1}{\Delta_{S}}(\gamma_{S}+\zeta_D+\Upsilon_{D}
\\
&
+\frac{M\gamma_{S,I}(\zeta_R+\frac{\E\Tr\BQ\FR\BQ\BD}{M})}{L\delta^2}  )+o(1)
 \\
 &=\frac{1}{\Delta_S}\{\gamma_{S}+\frac{1}{\Delta_S}[\frac{\gamma_{S,I} \E\Tr\BQ\BD\BQ\FR }{\delta^2 L}+\frac{\E\Tr\BQ\BD\BQ\BD}{M}
 \\
 &
 +\frac{M\gamma_{S,I}(\frac{\gamma_{S,I} \E\Tr\BQ\FR\BQ\FR }{\delta^2 L}+\frac{\E\Tr\BQ\FR\BQ\BD}{M})}{L\delta^2}]\}+o(1).
\end{aligned}
\end{equation}
Therefore, we turn to evaluate $\frac{\gamma_{S,I} \E\Tr\BQ\FR\BQ\FR }{\delta^2 L}+\frac{\E\Tr\BQ\FR\BQ\BD}{M}$ and $\frac{\gamma_{S,I} \E\Tr\BQ\FR\BQ\BD }{\delta^2 L}+\frac{\E\Tr\BQ\BD\BQ\BD}{M}$. To this end, we evaluate a more general form $P(\BM)=\frac{\gamma_{S,I}\E\Tr\BM\BQ\FR\BQ}{L\delta^2}+\frac{\E\Tr\BM\BQ\BD\BQ}{M}$ for any $\BM$ with $\limsup_{M \ge 1}\|\BM \|<\infty$ and define the auxiliary value $E(\BM)=\frac{\E\Tr\FR\BQ\BM\BQ}{L}$. By the Gaussian tools, we can set up a system of equations with respect to $E(\BM)$ and $P(\BM)$, which is given in~(\ref{eqs_second_eva}) in the middle of the next page, where $\gamma_{R,M}=\frac{\Tr\FR\BG\BM\BG}{L}$ and $\gamma_{D,M}=\frac{\Tr\BD\BG\BM\BG}{L}$. Therefore, we have
\begin{equation}
\begin{aligned}
&\chi+2\zeta_D+\Upsilon_D=\frac{1}{\Delta_{S}}(\gamma_{S}+\frac{P(\BD)+\frac{M\gamma_{S,I}P(\FR)}{L\delta^2}}{\Delta_{S}})
\\
&=\frac{\gamma_S}{\Delta_{S}}+(\frac{ [\bold{\Pi}^{-1}]_{(2)}\bold{q}_{RD}  }{\Delta_S^2}+\frac{\gamma_S}{\Delta_S}).
\end{aligned}
\end{equation}
\begin{figure*}
\begin{equation}
\label{eqs_second_eva}
\begin{bmatrix}
1-\frac{M\gamma_{S}\overline{\omega}^2\gamma_{R}}{L\delta^2} & -\frac{M\overline{\omega}^2 \gamma_{R}\gamma_{S,I} }{L\delta^2\Delta_S}- \frac{\overline{\omega}^2\gamma_{R,D}}{\Delta_S}
\\
  -\frac{\gamma_S\overline{\omega}^2(\frac{M\gamma_{S,I}\gamma_R}{L\delta^2}+\gamma_{R,D})}{\delta^2} &1-\frac{\overline{\omega}^2 (\frac{M\gamma_{S,I}\gamma_{R}}{L\delta^2}+ \gamma_{R,D})\gamma_{S,I} }{\delta^2\Delta_S}- \frac{\overline{\omega}^2(\frac{\gamma_{S,I}\gamma_{R,D}}{\delta^2}+\gamma_{D})}{\Delta_S}
\end{bmatrix}
\begin{bmatrix}
E(\BM) 
\\
P(\BM)
\end{bmatrix}
=
\begin{bmatrix}
\gamma_{R,M}
\\
\frac{\gamma_{S,I}\gamma_{R,M}}{\delta^2}+\frac{L\gamma_{D,M}}{M}
\end{bmatrix}+o(1).
\end{equation}
\hrulefill
\end{figure*}
Now we will evaluate $C_2$. By the definition of $\BQ$ in~(\ref{resol_def}), we have $\BQ\BH\BH^{H}=\bold{I}-\sigma^2\BQ$ and
\begin{equation}
\label{C_2_eva}
\begin{aligned}
& C_2 = \E [ \frac{N}{M} - \frac{\sigma^4}{M}\Tr\BQ^2 -\frac{2\sigma^2}{M}\Tr\BQ^2\BH\BA\BH^{H}]\E\Phi_{n}^{\BY,\BX,\BU} 
\\
&
+ o(1)
 =(\alpha-\frac{\sigma^4 C_{2,1}}{\beta}+C_{2,2})\E \Phi_{n}^{\BY,\BX,\BU} + o(1),
\end{aligned}
\end{equation}
where $C_{2,1}=\frac{\E\Tr\BQ^2}{L}$ can be evaluated by
\begin{equation}
\frac{\E\Tr\BQ^2}{L}= \gamma_{I}+\bold{p}_{I}^{T}\bold{\Pi}^{-1}\bold{q}_{I}+o(1),
 \end{equation}
 and $C_{2,2}= -\frac{2\sigma^2}{M}\Tr\BQ^2\BH\BA\BH^{H}=o(1)$. The detailed computation is similar to~\cite[Appendix G]{zhang2022asymptotic} and omitted here.

\subsection{Step 3: Convergence of $\E e^{\jmath\frac{\gamma_{n}^{\BW,\BY,\BX,\BU}-\overline{\gamma}_{n}}{\sqrt{V_n}}}$}
By~(\ref{fir_eva_W}),~(\ref{C_1_eva}), and~(\ref{C_2_eva}), we can obtain the derivative of the characteristic function
\begin{equation}
\label{diff_eq}
\begin{aligned}
&\frac{\partial  \Psi^{\BW,\BY,\BX,\BU}(u)}{\partial u}\xrightarrow[]{N  \xrightarrow[]{\alpha,\beta, \tau}\infty}
\\
&
 (\jmath \sqrt{nM}\overline{C}(\sigma^2) -u V_n )
\Psi^{\BW,\BY,\BX,\BU}(u)
\end{aligned}
\end{equation}
and conclude~(\ref{conv_cha}) by
\begin{equation}
\label{conver_cha}
\begin{aligned}
&
\Psi^{\BW,\BY,\BX,\BU}_{norm}(u)=\E e^{\jmath u \frac{\gamma_{n}-\overline{\gamma}_n}{\sqrt{V_n}}}
\\
&
=\Psi^{\BW,\BY,\BX,\BU}(\frac{u}{\sqrt{V_n}})e^{-\jmath u\frac{\overline{\gamma}_{n}}{\sqrt{V_n}}}
+o(1)
\\
&=e^{-\frac{u^2}{2}}
+o(1)
=e^{-\frac{u^2}{2}}
+o(1).
\end{aligned}
\end{equation}

\bibliographystyle{IEEEtran}
\bibliography{IEEEabrv,ref}
\end{document}